\documentclass[10pt]{IEEEtran}
\usepackage{cite}
\usepackage{amsmath,amssymb}
\usepackage{algorithmic}
\usepackage{graphicx}
\usepackage{tikz}
\usetikzlibrary{hobby,decorations.markings}
\usepackage{textcomp}
\usepackage{color}
\usepackage[english]{babel}
\usepackage{url}
\usepackage{multirow}
\usepackage{hhline}
\usepackage{booktabs}
\usepackage[linesnumbered,ruled,vlined,longend]{algorithm2e}

\usepackage[colorlinks = true,
linkcolor = blue,
urlcolor  = blue,
citecolor = blue,
anchorcolor = blue]{hyperref}

\newtheorem{mydef}{Definition}
\newtheorem{mylem}{Lemma}
\newtheorem{mypbm}{Problem}

\newtheorem{mycor}{Corollary}
\newtheorem{myprs}{Proposition}

\newcommand{\m}{\boldsymbol}

\title{Stabilizing Traffic Flow Via Autonomous Vehicles: A Less Conservative Approach}

\author{MirSaleh Bahavarnia$^\dagger$, \IEEEmembership{Member, IEEE} and Ahmad F. Taha$^\dagger$, \IEEEmembership{Member, IEEE}
	\thanks{$^\dagger$The authors are with the Department of Civil and Environmental Engineering, Vanderbilt University, 2201 West End Avenue, TN 37235, USA. Ahmad F. Taha is also affiliated with the Department of Electrical and Computer Engineering. This work was supported by the National Science Foundation (NSF) under Grant 2152450.} 
\thanks{Email addresses: \{mirsaleh.bahavarnia,ahmad.taha\}@vanderbilt.edu.}
}

\begin{document}

\maketitle

\begin{abstract}
This paper explores stabilizing traffic flow using a minimum number of autonomous vehicles (AVs) under control constraints. In contrast to most studies, we consider a \textit{heterogeneous} parameter setup scenario for human-driven vehicles (HVs) to reflect real-world differences in driving behavior. While current literature uses an $\mathcal{H}_{\infty}$-based sufficient condition to ensure the string stability of traffic flow, this often yields a conservative lower bound on the AV penetration rate to stabilize traffic flow. To reduce such conservativeness and obtain a \textit{less conservative} lower bound, we ensure the string stability of traffic flow by directly imposing the possession of no growing eigenmodes. We also systematically find a minimum number of required AVs and solve for the optimal control parameters via nonlinear optimization. We finally assess the intended conservativeness reduction via numerical simulations. Quantitatively, applying our algorithm to the homogeneous HV baseline in the literature (the result built upon an $\mathcal{H}_{\infty}$-based sufficient condition) reduces (improves) the AV penetration rate by 17.14\% while ensuring the string stability of traffic flow. We observe a trade-off between the stabilization/performance degradation and the number of utilized identical AVs. Quantitatively, our last numerical simulation corroborates that the AV penetration rate can be reduced by 61.54\% at the expense of 27.66\% higher position difference deviation from the equilibrium and a 92.47\% degradation in the real stability radius (RSR)---a metric to measure the stability robustness under the perturbation/uncertainty---associated with the aggregated linearized dynamics while ensuring the string stability of traffic flow. This trade-off helps engineers/operators make better traffic control decisions.
\end{abstract} 


\section{Introduction and Motivation} \label{sec:Intro}

\IEEEPARstart{S}{tabilization} of traffic flow via autonomous vehicles (AVs) has attracted increasing attention in the past two decades \cite{cui2017stabilizing,wang2017string,stern2018dissipation,wu2018stabilizing,monteil2019mathcal,zheng2020smoothing,wang2021controllability,wang2021optimal,giammarino2021traffic,li2022cooperative,10753487,hayat2025traffic,ameli2025design}. Undoubtedly, any improvement in the stabilization of traffic flow is equivalent to achieving less fuel consumption and emissions. To ensure the traffic flow's stability and effectively attenuate the perturbations originating from the collective behavior of the human-driven vehicles (HVs) \cite{wilson2011car}, one of the standard tools is the notion of \textit{string stability} \cite{swaroop1996string}. A traffic flow is called string stable if the perturbations do not amplify through it. Mathematically, string stability is ensured if the system does not possess any growing eigenmodes. It is noteworthy that the possession of no growing eigenmodes is equivalent to the non-positivity of the spectral abscissa---the greatest real part of the eigenmodes.

The direct imposition of the string stability---or its necessary and sufficient condition---could theoretically be challenging. Therefore, for simplicity, an $\mathcal{H}_{\infty}$-based \textit{sufficient condition} on the string stability has mainly been employed in the literature \cite{cui2017stabilizing,wang2017string,wu2018stabilizing,monteil2019mathcal,10753487}. In general, imposing such a sufficient condition could lead to conservative string stable traffic flows. Here, by conservativeness, we mean an unnecessary increase in the AV penetration rate. Later on, through numerical simulations, we will demonstrate that reducing such conservativeness can lead to the utilization of a smaller number of required AVs for traffic flow stabilization. For the specific case of a circular road, the only scenario in which such an $\mathcal{H}_{\infty}$-based sufficient condition on the string stability is equivalent to the exact string stability criterion is a limit case in which the number of vehicles tends to an extremely large value \cite{cui2017stabilizing,giammarino2021traffic}. In \cite{giammarino2021traffic}, a thorough study has been conducted to visualize the conservativeness caused by an $\mathcal{H}_{\infty}$-based sufficient condition on the string stability. Precisely, they have numerically verified that the fewer the number of vehicles, the more conservative such a condition becomes.

The aforementioned observations on the $\mathcal{H}_{\infty}$-based \textit{conservativeness} motivate the methods presented in this paper. To that end, we propose a direct approach to impose the string stability by directly incorporating the non-positivity of the spectral abscissa into the traffic flow stabilization algorithm. Moreover, we provide a more realistic scenario by \textit{(i)} considering heterogeneous HVs and \textit{(ii)} imposing the rational driving constraints (RDCs) and bounds on the control parameters, also known as box constraints (BCs).

Motivated by the conservative string stable results in the literature, one can pose the following question for a mixed vehicular platoon consisting of heterogeneous HVs and identical AVs:

\textit{Q: Can we attain a less conservative string stable traffic flow via a smaller number of identical AVs?}

We investigate the posed question throughout this paper. 

\textbf{Paper Contributions.} The paper's contributions are given as follows:
\begin{itemize}
    \item For a mixed vehicular platoon consisting of heterogeneous HVs and identical AVs, we ensure the string stability of traffic flow by directly imposing the possession of no growing eigenmodes, reducing the conservativeness, and obtaining a less conservative lower bound on the AV penetration rate compared to the $\mathcal{H}_{\infty}$-based alternative in the literature.
    \item We derive the conservative and more conservative---yet less computationally challenging---theoretical lower bounds on the AV penetration rate. Such a computational challenge in the former scenario arises from the heterogeneity of HVs. Compared to the literature,  we consider a more realistic heterogeneous HV scenario and propose an algorithm generating a less conservative string stable traffic flow via a smaller number of identical AVs.
    \item Given a fixed number of heterogeneous HVs, built upon the derived theoretical lower bounds, we systematically find a minimum number of required identical AVs and solve for the optimal control parameters via nonlinear optimization techniques.
    \item Through numerical simulations, we assess the effectiveness of the proposed less conservative approach compared to the $\mathcal{H}_{\infty}$-based alternative in the literature. Furthermore, we observe a trade-off between the stabilization/performance degradation and the number of utilized identical AVs. The information extracted from such a trade-off can benefit traffic control engineers/operators in traffic control decision-making tasks.
\end{itemize}

\textbf{Paper Notations.} We represent all the vectors and matrices in boldface notation. We denote the Laplace domain variable by $s$. We represent the sets of real, positive real, and complex numbers by $\mathbb{R}$, $\mathbb{R}_{+}$, and $\mathbb{C}$, respectively. We denote the left half plane by $\mathbb{C}^{-} := \{s \in \mathbb{C}: \Re(s) \le 0\}$. We symbolize the $n \times m$ zero and $n$-dimensional identity matrices by $\m O_{n \times m}$ and $\m I_n$, respectively. For a matrix $\m M$, we denote its Frobenius norm by $\|\m M\|_F$, which is defined as $\|\m M\|_F := \sqrt{\sum_{i,j} m_{ij}^2}$. For a square matrix $\m M$, we denote its spectral abscissa---the greatest real part of the matrix's spectrum---by $\rho[\m M]$. To show the real part and absolute value of a complex number $z$, we use $\Re(z)$ and $|z|$, respectively. We use $\imath$ to denote the imaginary unit $\sqrt{-1}$. We use $\cup$ and $\cap$ to show the set union and intersection, respectively. For a set $S$, the symbols $|S|$, $\inf S$, and $\min S$ denote the cardinality, infimum, and minimum of set $S$, respectively. We show the $d$-dimensional vector of all ones and all zeros by $\mathbf{1}_d$ and $\mathbf{0}_d$, respectively. We use $\lceil . \rceil$ to represent the ceiling function. For a scalar transfer function $T(s)$, we denote its $\mathcal{H}_{\infty}$ norm by $\|T(s)\|_{\infty}$, which is defined as $\|T(s)\|_{\infty} := {\sup_{\omega > 0}}~|T(\imath \omega)|$. Superscripts $(.)^{\mathrm{GM}}$ and $(.)^{\mathrm{AM}}$ denote the geometric and arithmetic means, respectively. We represent the element-wise Hadamard product by $\odot$. 

\textbf{Paper Structure.} The remainder of the paper is structured as follows: Section \ref{VDPS} details the mixed vehicular platoon dynamics---consisting of heterogeneous HV dynamics and identical AV dynamics---and formally formulates the problem statement. Section \ref{SAV} elaborates on ensuring less conservative string stability via identical AVs. Specifically, it presents three theoretical lower bounds on the AV penetration rate: 1) the conservative lower bound, 2) the more conservative lower bound, and 3) the less conservative lower bound. The conservative lower bounds are derived based on the sufficient string stability criterion, and the less conservative lower bound is derived based on the direct string stability criterion. Section \ref{Prcre} presents an algorithm to find the less conservative lower bound on the AV penetration rate via solving for optimal control parameters subject to string stability, RDCs, and BCs. Section \ref{Con} summarizes the paper. 

\section{Mixed Vehicular Platoon Dynamics} \label{VDPS}

We consider a circular road with a single lane, no ramps, and uniform road conditions \cite{cui2017stabilizing,wu2018stabilizing,10753487}. Tab. \ref{tab:my_label} summarizes traffic flow dynamics quantities. Let us assume the following ordering of the vehicles: vehicle $j+1$ precedes (leads) vehicle $j$ for $j \in \mathbb{N}_n$ (for $j = n$, vehicle $n+1$ is defined as vehicle $1$). In this paper, we limit our attention to the case of near-equilibrium flow (i.e., local stabilization). Vehicles can be categorized into two types: \textit{(i)} heterogeneous HVs and \textit{(ii)} identical AVs. Then, we accordingly have $\mathbb{N}_n = \mathcal{I}_{\mathrm{HV}} \cup \mathcal{I}_{\mathrm{AV}}$ with $|\mathcal{I}_{\mathrm{HV}}| = n-m$ and $|\mathcal{I}_{\mathrm{AV}}| = m$. Fig. \ref{fig0} depicts a schematic of a mixed vehicular platoon consisting of $4$ heterogeneous HVs and $2$ identical AVs.

\begin{table}[!t]
    \centering
    \caption{Summary of traffic flow dynamics quantities \cite{cui2017stabilizing,wu2018stabilizing,10753487}}
    \begin{tabular}{@{}c p{2.3in}@{}}
    \toprule
        Notation & Definition \\
    \midrule
        $L$ & Road length \\
        $n$ & Number of vehicles \\
        $m$ & Number of identical AVs\\
        $n-m$ & Number of heterogeneous HVs\\
        $\gamma := \frac{m}{n}$ & AV penetration rate\\
        $\mathbb{N}_n$ & Index set associated with vehicles: $\{1,\dots,n\}$\\
        $\mathcal{I}_{\mathrm{AV}}$ & Index set associated with identical AVs\\
        $\mathcal{I}_{\mathrm{HV}}$ & Index set associated with heterogeneous HVs\\
        $t$ & Time\\
        $x_j(t)$ & $\mathrm{Position~along~the~road~(defined~modulo}~L\mathrm{)}$ \newline of the~$j$-th~$\mathrm{vehicle~at~time}~t$\\
        $v_j(t) := \dot{x}_j(t)$ & Velocity of the~$j$-th~$\mathrm{vehicle~at~time}~t$\\
        $a_j(t) := \ddot{x}_j(t)$ & Acceleration of the~$j$-th~$\mathrm{vehicle~at~time}~t$\\
        $h_j(t)$ & Spacing of the~$j$-th~$\mathrm{vehicle~at~time}~t$: \newline $x_{j+1}(t)-x_j(t)$\\
        $\dot{h}_j(t)$ & Relative velocity of the~$j$-th~vehicle at \newline $\mathrm{time}~t$: $\dot{x}_{j+1}(t)-\dot{x}_j(t)$\\
        $\m x_{\mathrm{eq}}(t) \in \mathbb{R}^n$ & Equilibrium position\\
        $\m v_{\mathrm{eq}}(t) \in \mathbb{R}^n$ & Equilibrium velocity\\
        $\m a_{\mathrm{eq}} = \mathbf{0}_n$ & Equilibrium acceleration\\
        $\m h_{\mathrm{eq}} = \frac{L}{n} \mathbf{1}_n$ & Equilibrium spacing\\
        $y_j(t)$ & Infinitesimal position difference deviation of \newline the~$j$-th~$\mathrm{vehicle~at~time}~t$: $x_j(t)-x_{\mathrm{eq},j}(t)$\\
        $u_j(t)$ & Infinitesimal velocity difference deviation of \newline the~$j$-th~$\mathrm{vehicle~at~time}~t$: $v_j(t)-v_{\mathrm{eq},j}(t)$\\
    \bottomrule
    \end{tabular}
    \label{tab:my_label}
\end{table}

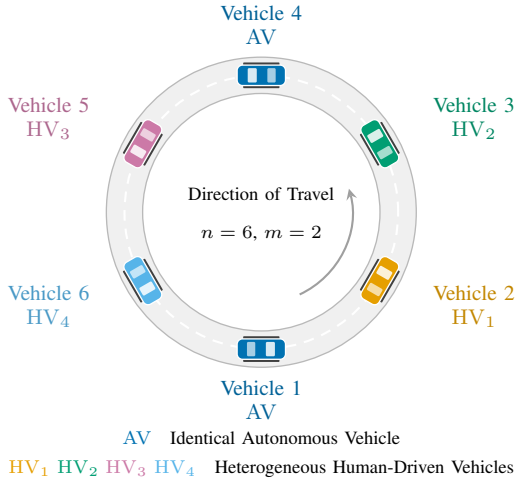
\begin{figure}[t]
    \centering
\begin{tikzpicture}[x=1cm,y=1cm,font=\footnotesize]
\definecolor{avblue}{RGB}{0,114,178}      
\definecolor{hvorange}{RGB}{230,159,0}    
\definecolor{hvgreen}{RGB}{0,158,115}     
\definecolor{hvpurple}{RGB}{204,121,167}  
\definecolor{hvskyblue}{RGB}{86,180,233}  

\fill[gray!12,even odd rule] (0,0) circle (2.05) (0,0) circle (1.57);
\draw[gray!55,line width=0.5pt] (0,0) circle (2.05);
\draw[gray!55,line width=0.5pt] (0,0) circle (1.57);
\draw[white,line width=0.8pt,dash pattern=on 4pt off 3pt] (0,0) circle (1.81);

\foreach \j/\ang/\col/\kind/\anch in {
  1/-90/avblue/\mathrm{AV}/center,
  2/-30/hvorange/\mathrm{HV}_{1}/west,
  3/30/hvgreen/\mathrm{HV}_{2}/west,
  4/90/avblue/\mathrm{AV}/center,
  5/150/hvpurple/\mathrm{HV}_{3}/east,
  6/210/hvskyblue/\mathrm{HV}_{4}/east%
}{
  \begin{scope}[shift={(\ang:1.81)},rotate=\ang+90]
    \fill[black!75,rounded corners=0.3pt] (-0.24,-0.17) rectangle (0.24,0.17);
    \fill[\col,draw=white,line width=0.5pt,rounded corners=2pt] (-0.32,-0.135) rectangle (0.32,0.135);
    \fill[white!85!\col,rounded corners=0.5pt] (0.06,-0.095) rectangle (0.18,0.095);
    \fill[white!65!\col,rounded corners=0.5pt] (-0.19,-0.095) rectangle (-0.08,0.095);
  \end{scope}
  \node[align=center,anchor=\anch,text=\col!85!black] at (\ang:2.48) {Vehicle \j\\$\kind$};
}

\draw[->,>=stealth,gray!75,line width=0.8pt] (-65:1.20) arc (-65:15:1.20);
\node[align=center,font=\scriptsize] at (0,0.25) {Direction of Travel};
\node[align=center,font=\scriptsize] at (0,-0.28) {$n=6$, $m=2$};

\node[align=center,font=\scriptsize] at (0,-3.2) {
  \textcolor{avblue}{$\mathrm{AV}$}\quad Identical Autonomous Vehicle\\[3pt]
  \textcolor{hvorange}{$\mathrm{HV}_1$}\ \textcolor{hvgreen}{$\mathrm{HV}_2$}\ \textcolor{hvpurple}{$\mathrm{HV}_3$}\ \textcolor{hvskyblue}{$\mathrm{HV}_4$}\quad Heterogeneous Human-Driven Vehicles
};
\end{tikzpicture}
    \caption{A mixed platoon configuration with $n = 6$ vehicles and $m = 2$ identical AVs. The four heterogeneous HVs are assigned unique colors and index notation ($\mathrm{HV}_1$ to $\mathrm{HV}_4$) to model parametric variation, where $\mathcal{I}_{\mathrm{HV}} = \{2,3,5,6\}$ and $\mathcal{I}_{\mathrm{AV}} = \{1,4\}$.}
    \label{fig0}
\end{figure}

\subsection{Mixed vehicular platoon dynamics}

In this section, to derive theoretical stability results, we linearize the nonlinear dynamics around the equilibrium. We emphasize that the theoretical derivations will only be valid under small perturbations, as when the spacing/velocity deviations are large enough, the higher-order terms in the Taylor expansion can no longer be ignored.

\subsubsection{Heterogeneous HV dynamics}
For each heterogeneous HV, we consider the following second-order car-following dynamics:
\begin{align} \label{HVE}
    \ddot{x}_j &= f_j(h_j,\dot{h}_j,v_j),~j \in \mathcal{I}_{\mathrm{HV}}.
\end{align}Among many examples, one crucial example of car-following dynamics describable by \eqref{HVE} is the \textit{optimal-velocity-follow-the-leader} (OV-FTL) model \cite{bando1994structure,cui2017stabilizing}.

Utilizing the quantities detailed in Tab. \ref{tab:my_label} and considering dynamics \eqref{HVE} under small perturbations from the equilibrium flow, we get the following dynamics:
\begin{align} \label{LHVE}
    & \ddot{y}_j = \alpha_{j1} (y_{j+1}-y_j) - \alpha_{j2} u_j + \alpha_{j3} u_{j+1},~j \in \mathcal{I}_{\mathrm{HV}},\\
    &\alpha_{j1} = \frac{\partial f_j}{\partial h_j} \bigg{|}_{\mathrm{eq}},~\alpha_{j2} = \frac{\partial f_j}{\partial \dot{h}_j}\bigg{|}_{\mathrm{eq}}-\frac{\partial f_j}{\partial v_j}\bigg{|}_{\mathrm{eq}},~\alpha_{j3} = \frac{\partial f_j}{\partial \dot{h}_j}\bigg{|}_{\mathrm{eq}}.\notag
\end{align}   
For linearized dynamics \eqref{LHVE}, the following standard assumptions hold \cite{cui2017stabilizing}: the acceleration of vehicle $j$ is reduced when \textit{(i)} the spacing $h_j$ decreases, \textit{(ii)} the relative velocity $\dot{h}_j$ decreases, or \textit{(iii)} the vehicle's velocity $v_j$ increases. Such risk aversion criteria imply the following rational driving constraints (RDCs) \cite{wilson2011car}: $\alpha_{j1} > 0$, $\alpha_{j2} > \alpha_{j3}$, and $\alpha_{j3} > 0$. To determine the poles associated with linearized dynamics \eqref{LHVE}, one can take the Laplace transform of \eqref{LHVE} leading to the following transfer function $T_j(s) := \frac{Y_j(s)}{Y_{j+1}(s)}$:
\begin{align} \label{TFHV}
    & T_j(s) = F(s;\m \alpha_j) = \frac{\alpha_{j3} s + \alpha_{j1}}{s^2 + \alpha_{j2} s + \alpha_{j1}},~j \in \mathcal{I}_{\mathrm{HV}},
\end{align}
where $\m \alpha_j := \begin{bmatrix}
    \alpha_{j1} & \alpha_{j2} & \alpha_{j3}
\end{bmatrix}^\top$ denotes the system parameter vector for $j \in \mathcal{I}_{\mathrm{HV}}$. The Hurwitz stability of transfer function $F(s;\m \alpha_j)$ is equivalent to the simultaneous satisfaction of $\alpha_{j1} > 0$ and $\alpha_{j2} > 0$.

\subsubsection{Identical AV dynamics}

Similarly, for each AV, we consider the following second-order car-following dynamics:
\begin{align} \label{AVE}
    \ddot{x}_j &= g(h_j,\dot{h}_j,v_j),~j \in \mathcal{I}_{\mathrm{AV}}.
\end{align}Considering dynamics \eqref{AVE} under small perturbations from the equilibrium flow, we get the following linearized dynamics:
{\begin{align} \label{LAVE}
    & \ddot{y}_j = \beta_1 (y_{j+1}-y_j) - \beta_2 u_j + \beta_3 u_{j+1},~j \in \mathcal{I}_{\mathrm{AV}},\\
    &\beta_1 = \frac{\partial g}{\partial h_j} \bigg{|}_{\mathrm{eq}},~\beta_2 = \frac{\partial g}{\partial \dot{h}_j}\bigg{|}_{\mathrm{eq}}-\frac{\partial g}{\partial v_j}\bigg{|}_{\mathrm{eq}},~\beta_3 = \frac{\partial g}{\partial \dot{h}_j}\bigg{|}_{\mathrm{eq}}, \notag
\end{align}}Likewise, for linearized dynamics \eqref{LAVE}, the following standard assumptions hold \cite{cui2017stabilizing}: the acceleration of vehicle $j$ is reduced when \textit{(i)} the spacing $h_j$ decreases, \textit{(ii)} the relative velocity $\dot{h}_j$ decreases, or \textit{(iii)} the vehicle's velocity $v_j$ increases. Such risk aversion criteria imply the following RDCs \cite{wilson2011car}:
\begin{align} \label{RDCs}
    & \beta_1 > 0,~\beta_2 - \beta_3 > 0,~\beta_3 > 0.
\end{align}
Similarly, to determine the poles associated with linearized dynamics \eqref{LAVE}, one can take the Laplace transform of \eqref{LAVE} leading to the following transfer function $T_j(s) := \frac{Y_j(s)}{Y_{j+1}(s)}$:
\begin{align} \label{TFAV}
    & T_j(s) = G(s;\m \beta) = \frac{\beta_3 s + \beta_1}{s^2 + \beta_2 s + \beta_1},~j \in \mathcal{I}_{\mathrm{AV}},
\end{align}
where $\m \beta := \begin{bmatrix}
    \beta_1 & \beta_2 & \beta_3
\end{bmatrix}^\top$ denotes the control parameter vector. The Hurwitz stability of transfer function $G(s;\m \beta)$ is equivalent to the simultaneous satisfaction of $\beta_{1,2} > 0$.

It is noteworthy that in the case of ``non-identical" AVs, one can have more free control parameters to attain a broader set of control objectives practically. However, the theoretical derivations in the current work heavily depend on the simplicity of the case of identical AVs, i.e., having only $3$ free control parameters, and the stability analysis in the case of non-identical AVs is more complex. Furthermore, in the case of non-identical AVs, searching for a feasible stabilizing set of control parameters will accordingly become more complicated due to the increased time complexity of the nonlinear optimization solver.

One can obtain a state-space representation for the aggregated dynamics associated with the mixed vehicular platoon consisting of heterogeneous HVs and identical AVs by combining the heterogeneous HVs' linearized dynamics \eqref{LHVE} and the identical AVs' linearized dynamics \eqref{LAVE}. To that end, defining $\m r := \begin{bmatrix}
    y_1 & \cdots & y_n & u_1 & \cdots & u_n
\end{bmatrix}^\top$ as a state vector, the aggregated linearized dynamics can equivalently be rewritten via the following $\m \beta$-dependent closed-loop state-space representation:
\begin{subequations} \label{LD}
\begin{align}
    & \dot{\m r}(t) = \underbrace{\begin{bmatrix}
        \m O_{n \times n} & \m I_n\\\m A(\m \beta) & \m B(\m \beta)
    \end{bmatrix}}_{\m M(\m \beta)} \m r(t),\\
    & \m A(\m \beta) : \begin{cases}
       a_{jj} = -\alpha_{j1},~a_{j(j+1)} = \alpha_{j1} & j \in \mathcal{I}_{\mathrm{HV}}\\
       a_{jj} = -\beta_{1},~a_{j(j+1)} = \beta_{1} & j \in \mathcal{I}_{\mathrm{AV}}
    \end{cases},\\
    & \m B(\m \beta): \begin{cases}
       b_{jj} = -\alpha_{j2},~b_{j(j+1)} = \alpha_{j3} & j \in \mathcal{I}_{\mathrm{HV}}\\
       b_{jj} = -\beta_{2},~b_{j(j+1)} = \beta_{3} & j \in \mathcal{I}_{\mathrm{AV}}
    \end{cases},
\end{align}    
\end{subequations}where the $\m \beta$-dependent matrices $\m A(\m \beta) \in \mathbb{R}^{n \times n}$ and $\m B(\m \beta) \in \mathbb{R}^{n \times n}$ are defined based on their nonzero elements. Also, note that the subscript $j+1$ for $j=n$ is treated as $1$.

\subsection{Problem statement}

\subsubsection{String stability} A mixed vehicular platoon with linearized dynamics \eqref{LD} is said to be \textit{string stable} if infinitesimal perturbations do not amplify and the system remains close to the equilibrium \cite{wilson2011car,cui2017stabilizing}. The formal mathematical definition of string stability can be expressed as follows:
\begin{mydef}[String stability] \label{def1}
    A mixed vehicular platoon with linearized dynamics \eqref{LD} is said to be \textit{string stable} if all of its eigenmodes lie in the left half plane $\mathbb{C}^{-} := \{s \in \mathbb{C}: \Re(s) \le 0\}$.
\end{mydef}
Remarkably, we emphasize that we only have control over AVs, as there is no direct control over the behavior of HVs. Then, considering the linearized dynamics \eqref{LD}, later on, we will search for optimal control parameters to ensure the string stability.

\subsubsection{Box constraints} Let us assume the following bounds on the control parameter vector $\m \beta$ (also known as box constraints (BCs)):
\begin{align} \label{LUB}
    &\beta_1^l \le \beta_1 \le \beta_1^u,~\beta_2^l \le \beta_2 \le \beta_2^u,~\beta_3^l \le \beta_3 \le \beta_3^u.
\end{align}    
Note that $0 < \beta_i^l$ holds for all $i \in \{1,2,3\}$. For brevity, we will use notations $\m \beta^{l} := \begin{bmatrix}
    \beta_1^{l} & \beta_2^{l} & \beta_3^{l}
\end{bmatrix}^\top$ and $\m \beta^{u} := \begin{bmatrix}
    \beta_1^{u} & \beta_2^{u} & \beta_3^{u}
\end{bmatrix}^\top$ where necessary. The rationale behind the incorporation of the box constraints is that we can realize and implement the physically realizable controllers through these constraints \cite{wu2018stabilizing}. Note that actuators and sensors have hard physical limits (e.g., maximum acceleration or communication delays). Also, it is remarkable that we will not impose any additional conservativeness through the incorporation of the BCs \eqref{LUB} as our control parameters' parameterizations (in Proposition \ref{Propo1} derived later on in Section \ref{IIIB}) preserve the exact feasible stabilizing region.

Throughout this paper, we mathematically investigate the following problem:
\begin{mypbm} \label{Prob1}
    Given a mixed vehicular platoon represented by linearized dynamics \eqref{LD}, the RDCs \eqref{RDCs}, and the BCs \eqref{LUB}, find the optimal $\m \beta^{\ast}$ for which traffic flow can be stabilized with an optimally minimal AV penetration rate.
\end{mypbm}

\section{Less Conservative String Stability} \label{SAV}
This section is comprised of the following main parts: \textit{(i)} string stability criteria: 1) direct string stability criterion and 2) sufficient string stability criterion, \textit{(ii)} parameterized incorporation of the RDCs and the BCs, and \textit{(iii)} theoretical lower bounds on the AV penetration rate $\gamma$: 1) conservative lower bound, 2) more conservative lower bound, and 3) less conservative lower bound. 

\subsection{String stability criteria}

We consider two string stability criteria in this section: 1) the direct string stability criterion and 2) the sufficient string stability criterion. The former will be used to find the optimal control parameter vector $\m \beta^{\ast}$ while the latter will be used for deriving the conservative lower bounds on the AV penetration rate $\gamma$.

\subsubsection{Direct string stability criterion}
According to Definition \ref{def1}, string stability of \eqref{LD} reduces to the following direct string stability criterion:
\begin{align} \label{SpAb}
    \boxed{\rho[\m M(\m \beta)] \le 0.}
\end{align}
Throughout this paper, we will use \eqref{SpAb} to directly ensure the string stability. To derive the theoretical lower bounds on the AV penetration rate $\gamma$, we present an $\mathcal{H}_{\infty}$-based sufficient condition on the string stability---similar to the $\mathcal{H}_{\infty}$-based sufficient conditions derived by \cite{cui2017stabilizing,wu2018stabilizing,10753487} for the case of traffic flow with homogeneous HVs. We will then compare how conservative the derived theoretical lower bounds are.

\subsubsection{Sufficient string stability criterion} Considering \eqref{TFHV} and \eqref{TFAV}, and according to the periodicity of the circular road, we can solve
\begin{align} \label{PCR}
    \prod_{j \in \mathbb{N}_n} T_j(s) &=  G(s;\m \beta)^m \prod_{j \in \mathcal{I}_{\mathrm{HV}}} F(s;\m \alpha_j) = 1,
\end{align}
for the $2n$ roots that are the eigenmodes of the linearized dynamics \eqref{LD}. The $2n$ roots of \eqref{PCR} lie on a curve  $\mathcal{C} := \{s \in \mathbb{C}: |G(s;\m \beta)|^{\gamma} |F^{\mathrm{GM}}(s)|^{1-\gamma} = 1\}$ where $\gamma:= \frac{m}{n}$ denotes the AV penetration rate and $F^{\mathrm{GM}}(s) := [\prod_{j \in \mathcal{I}_{\mathrm{HV}}} F(s;\m \alpha_j)]^{\frac{1}{n-m}}$ represents the geometric mean of the heterogeneous HV transfer functions $F(s;\m \alpha_j)$. 

A sufficient condition on the string stability can be formulated as $\mathcal{C} \subset \mathbb{C}^{-}$. We define the following fractional-order transfer function:
\begin{align*}
    & H_{\gamma}(s) := G(s;\m \beta)^{\gamma} F^{\mathrm{GM}}(s)^{1-\gamma}.
\end{align*}To sufficiently ensure the string stability of the linearized dynamics \eqref{LD}, it suffices to consider $\mathcal{C} \subset \mathbb{C}^{-}$ and equivalently impose $|H_{\gamma}(\imath \omega)| \le 1$ for all $\omega \in \mathbb{R}$ (i.e., $\|H_{\gamma}(s)\|_{\infty} \le 1$) which is equivalent to the following sufficient string stability criterion:
\begin{align} \label{logEA}
    & \boxed{\gamma D_{\m \beta}(\omega) + (1-\gamma) D^{\mathrm{AM}}(\omega)  \le 0,~\forall \omega \in \mathbb{R},}\\
    & D_{\m \beta}(\omega) := \ln (|G(\imath \omega;\m \beta)|) = \frac{1}{2} \ln \bigg (\frac{\beta_3^2 \omega^2 + \beta_1^2}{\beta_2^2 \omega^2 + (\omega^2-\beta_1)^2} \bigg),\notag\\
    & D^{\mathrm{AM}}(\omega):= \frac{\sum_{j \in \mathcal{I}_{\mathrm{HV}}} D_{\m \alpha_j}(\omega)}{n-m}, \notag\\
    & D_{\m \alpha_j}(\omega) := \ln (|F(\imath \omega;\m \alpha_j)|) = \frac{1}{2} \ln\!\bigg (\frac{\alpha_{j3}^2 \omega^2 + \alpha_{j1}^2}{\alpha_{j2}^2 \omega^2 + (\omega^2-\alpha_{j1})^2} \bigg). \notag
\end{align}
Regardless of the values of $\m \beta$ and $\m \alpha_j$, \eqref{logEA} holds for $\omega = 0$. Also, since $(-\omega)^2 = \omega^2$ holds, w.l.o.g., we only consider the case of $\omega \ge 0$. Then, from now on, we assume $\omega > 0$.

Defining $\Delta_{\m \beta}:= -2\beta_1+\beta_2^2-\beta_3^2$ and $\Delta_{\m \alpha_j} := -2\alpha_{j1} + \alpha_{j2}^2-\alpha_{j3}^2$, it can be verified that
{\begin{subequations} \label{Deltas}
\begin{align}
    & D_{\m \beta}(\omega) < 0,~\forall \omega \in \mathbb{R}_{+} \iff \Delta_{\m \beta} \ge 0,\\
    & D_{\m \alpha_j}(\omega) < 0,~\forall \omega \in \mathbb{R}_{+} \iff \Delta_{\m \alpha_j} \ge 0,
\end{align}
\end{subequations}}hold. In the case of $\Delta_{\m \alpha_j} < 0$, we also have
\begin{align} \label{SignDforDeltaalpha}
    \mathrm{If}~\Delta_{\m \alpha_j} < 0 &: \begin{cases}
        D_{\m \alpha_j}(\omega) > 0, & \omega \in \Big(0,\sqrt{-\Delta_{\m \alpha_j}}\Big)\\
        D_{\m \alpha_j}(\omega) = 0, & \omega = \sqrt{-\Delta_{\m \alpha_j}}\\
        D_{\m \alpha_j}(\omega) < 0, & \omega \in \Big(\sqrt{-\Delta_{\m \alpha_j}},\infty \Big)
    \end{cases}.
\end{align}
The sufficient string stability criterion \eqref{logEA} along with \eqref{Deltas} and \eqref{SignDforDeltaalpha} will be utilized later on for deriving the theoretical lower bounds on the AV penetration rate $\gamma$.

\subsection{Parameterized Incorporation of the RDCs and the BCs} \label{IIIB}
In this section, we systematically parameterize the set of control parameters $\begin{bmatrix}
    \beta_1 & \beta_2 & \beta_3
\end{bmatrix}$ satisfying the RDCs \eqref{RDCs} and the BCs \eqref{LUB}. 

Let us define the following notations: $\mathcal{B}_1 := \{\m \beta \in \mathbb{R}_{+}^3: \eqref{RDCs} \mathrm{~holds~for~}\m \beta\}$ and $\mathcal{B}_2 := \{\m \beta \in \mathbb{R}_{+}^3: \eqref{LUB} \mathrm{~holds~for~}\m \beta\}$. The set of control parameters $\begin{bmatrix}
    \beta_1 & \beta_2 & \beta_3
\end{bmatrix}$ satisfying the RDCs \eqref{RDCs}, i.e., $\mathcal{B}_1$ can be parameterized via the parameters $\begin{bmatrix}
    p & q & r
\end{bmatrix}$ as
\begin{align} \label{pars}
    \beta_3(p) = p,~\beta_2(p,q) = p+q,~\beta_1(r) = r,
\end{align}
where $p > 0$, $q>0$, and $r > 0$ hold. In the following proposition, we systematically incorporate the BCs \eqref{LUB} into the parameterization \eqref{pars}.

\begin{myprs} \label{Propo1}
    The parameters $\begin{bmatrix}
    p & q & r
\end{bmatrix}$ in \eqref{pars} satisfying the BCs \eqref{LUB} can be parameterized via the parameters $\begin{bmatrix}
    \psi_1 & \psi_2 & \psi_3
\end{bmatrix}$ as
\begin{subequations} \label{pqrpar}
    \begin{align}
        & p = \texttt{p}(\psi_1) = (1-\psi_1)p^l + \psi_1 p^u,\\
        & q = \texttt{q}(\psi_1,\psi_2) = (1-\psi_2)q_{\psi_1}^l + \psi_2 q_{\psi_1}^u,\\
        & r = \texttt{r}(\psi_3) = (1-\psi_3)r^l + \psi_3 r^u,\\
        & p^l = \max \{\epsilon,\beta_3^l\},~p^u=\min \{\beta_3^u,\beta_2^u-\epsilon\},\\
        & q_{\psi_1}^l = \max \{\epsilon,\beta_2^l-\texttt{p}(\psi_1)\},~q_{\psi_1}^u=\beta_2^u - \texttt{p}(\psi_1),\\
        & r^l = \max \{\epsilon,\beta_1^l\},~r^u = \beta_1^u,
    \end{align}
\end{subequations}where $\psi_i \in [0,1]$ holds for all $i \in \{1,2,3\}$ and $\epsilon > 0$ is an infinitesimal value. Also, $\beta_3^u \ge \epsilon$, $\beta_2^u \ge \max \{2\epsilon,\beta_3^l+\epsilon\}$, and $\beta_1^u \ge \epsilon$ necessarily hold for $\m \beta^l$ and $\m \beta^u$.  
\end{myprs}
\begin{IEEEproof}
    See Appendix \ref{App1}.
\end{IEEEproof}
In \eqref{pqrpar}, since $\psi_i \in [0,1]$ holds for all $i \in \{1,2,3\}$, we can choose the form of $\psi_i$s via an arbitrary sigmoid function, e.g., the logistic function $\phi(\tau) = 1/(1+e^{-\zeta \tau})$ where $\zeta > 0$ denotes the growth rate. Any sigmoid function enables a one-to-one mapping between $(-\infty,\infty)$ and $(0,1)$ for parameterized optimization purposes with real-valued optimization variables. The following corollary immediately results from Proposition \ref{Propo1} and systematically parameterizes the set of control parameters $\begin{bmatrix}
        \beta_1 & \beta_2 & \beta_3
    \end{bmatrix}$ satisfying the RDCs \eqref{RDCs} and the BCs \eqref{LUB}.

\begin{mycor} \label{Cor1}
    Control parameters $\begin{bmatrix}
        \beta_1 & \beta_2 & \beta_3
    \end{bmatrix}$ satisfying the RDCs \eqref{RDCs} and the BCs \eqref{LUB} (any member of the set $\mathcal{B}_1 \cap \mathcal{B}_2$) can be parameterized via parameters $\begin{bmatrix}
        \theta_1 & \theta_2 & \theta_3
    \end{bmatrix}$ as
    \begin{subequations} \label{betap}
        \begin{align}
            & \beta_3(\theta_1) = \texttt{p}(\phi(\theta_1)),\\
            & \beta_2(\theta_1,\theta_2) = \texttt{p}(\phi(\theta_1)) + \texttt{q}(\phi(\theta_1),\phi(\theta_2)),\\
            & \beta_1(\theta_3) = \texttt{r}(\phi(\theta_3)),
        \end{align}
    \end{subequations}
    where $\theta_i \in \mathbb{R}$ holds for all $i \in \{1,2,3\}$, $\phi(.)$ denotes the logistic function, and $\texttt{p}(.)$, $\texttt{q}(.,.)$, and $\texttt{r}(.)$ represent the same functions expressed in \eqref{pqrpar}.
\end{mycor}

We will use the following notations for parameterized optimization purposes in the sequel: $\m \theta := \begin{bmatrix}
        \theta_1 & \theta_2 & \theta_3
    \end{bmatrix}^\top$, and $\m \beta(\m \theta) := \begin{bmatrix}
        \beta_1(\theta_3) & \beta_2(\theta_1,\theta_2) & \beta_3(\theta_1)
    \end{bmatrix}^\top$. The rationale behind the parameterized control parameter vector $\m \beta(\m \theta)$ characterized by \eqref{betap} in Corollary \ref{Cor1} is that we later on utilize it to facilitate solving a cast nonlinear optimization for a feasible solution $\m \theta^{\ast}$ in Section \ref{C3} to compute a less conservative lower bound on the AV penetration rate, namely $\gamma^{\ast}$. Moreover, we highlight that such a parameterization does not limit the feasible solution space, as throughout the parameterization process, we have equivalently (without imposing any simplifying assumptions) and systematically incorporated the RDCs \eqref{RDCs} and the BCs \eqref{LUB}.

\subsection{Lower bounds on the AV penetration rate} 

This section presents three theoretical lower bounds on the AV penetration rate $\gamma$: 1) the conservative lower bound, 2) the more conservative lower bound, and 3) the less conservative lower bound. The conservative lower bounds are derived based on the sufficient string stability criterion, and the less conservative lower bound is derived based on the direct string stability criterion.

\subsubsection{Conservative lower bound on the AV penetration rate} 

The sufficient string stability criterion \eqref{logEA} can equivalently be rearranged as
\begin{subequations} \label{SuffCon}
\begin{align}
    & \Delta_{\m \beta} \ge 0, \label{SuffCon1}\\
    & \gamma \ge \overbrace{\frac{1}{1+\underbrace{\inf \Big \{ \frac{-D_{\m \beta}(\omega)}{D^{\mathrm{AM}}(\omega)}: \omega > 0, D^{\mathrm{AM}}(\omega)>0 \Big \}}_{J^{\ast}(\m \beta)}}}^{K^{\ast}(\m \beta)}. \label{SuffCon2}
\end{align}    
\end{subequations}
Unlike the case of homogeneous HVs considered in \cite{wu2018stabilizing,10753487}, optimizing the lower bound $K^{\ast}(\m \beta)$ for $\m  \beta$ is more computationally challenging. That necessitates the derivation of a theoretical lower bound on $K^{\ast}(\m \beta)$ in the sequel. Built upon the following lemma, we next derive a lower bound on $K^{\ast}(\m \beta)$, namely $\underline{K}(\m \beta)$.
\begin{mylem} \label{Lemma1}
    If $\Delta_{\m \alpha_j} < 0$ holds for some $j \in \mathcal{I}_{\mathrm{HV}}$, then $D_{\m \alpha_j}(\omega)$ takes its global maximum value at
    \begin{align} \label{tilomeg}
        & \tilde{\omega}_j = \frac{\alpha_{j1}}{\alpha_{j3}}\sqrt{\sqrt{1-\bigg(\frac{\alpha_{j3}}{\alpha_{j1}}\bigg)^2\Delta_{\m \alpha_j}}-1}.
    \end{align}
Moreover, the global maximum value is
\begin{align} \label{GMV}
    & D_{\m \alpha_j}(\tilde{\omega}_j) = \frac{1}{2} \ln \Bigg( \frac{\alpha_{j1}^2}{\alpha_{j1}^2-{\tilde{\omega}_j}^4}\Bigg).
\end{align}
\end{mylem}
\begin{IEEEproof}
    See Appendix \ref{App2}.
\end{IEEEproof}

\begin{myprs} \label{Propo2}
    If \eqref{SuffCon1} holds and $\Delta_{\m \alpha_j} < 0$ holds for some $j \in \mathcal{I}_{\mathrm{HV}}$, then we get the following lower bound on $K^{\ast}(\m \beta)$:
    \begin{align} \label{LowB}
        & K^{\ast}(\m \beta) \ge \overbrace{\frac{1}{1+ \underset{j \in \tilde{\mathcal{I}}^{+}_{\mathrm{HV}}}{\min} \frac{-D_{\m \beta}(\tilde{\omega}_j)}{D^{\mathrm{AM}}(\tilde{\omega}_j)}}}^{\underline{K}(\m \beta)},
    \end{align}
    where $\tilde{\mathcal{I}}^{+}_{\mathrm{HV}} := \{j \in \mathcal{I}_{\mathrm{HV}}: \Delta_{\m \alpha_j} < 0, D^{\mathrm{AM}}(\tilde{\omega}_j) > 0\}$ and $\tilde{\omega}_j$ can be computed via \eqref{tilomeg}.
\end{myprs}
\begin{IEEEproof}
See Appendix \ref{App3}.
\end{IEEEproof}
According to \eqref{SuffCon}, we state the following corollary regarding the conservative lower bound on the AV penetration rate $\gamma$.
\begin{mycor} \label{Cor2}
    If \eqref{SuffCon1} holds and $\Delta_{\m \alpha_j} < 0$ holds for some $j \in \mathcal{I}_{\mathrm{HV}}$, then we get the following conservative lower bound on the AV penetration rate:
    \begin{align} \label{gKss}
        & \boxed{\gamma \ge K^{\ast \ast}}
    \end{align}
    where $K^{\ast \ast} := \min \{K^{\ast}(\m \beta): \m \beta \in \mathcal{B}_1 \cap \mathcal{B}_2, \Delta_{\m \beta} \ge 0\}$.
\end{mycor}
\subsubsection{More conservative lower bound on the AV penetration rate} Although the search space associated with $K^{\ast \ast}$ in \eqref{gKss} can identically be cast via \cite[Corollary 1]{10753487}, the computational difficulty associated with $K^{\ast \ast}$ obstructs the exact computation of $K^{\ast}(\m \beta)$. But, we can find an upper bound on $K^{\ast \ast}$, serving as a more conservative yet less computationally challenging lower bound on the AV penetration rate $\gamma$ via the following proposition.
\begin{myprs} \label{Propo3}
    If \eqref{SuffCon1} holds and $\Delta_{\m \alpha_j} < 0$ holds for some $j \in \mathcal{I}_{\mathrm{HV}}$, then we get $K^{\ast}(\hat{\m \beta}) \ge K^{\ast \ast}$ where
    \begin{align} \label{hatB}
        & \hat{\m \beta} := \arg \min \{\underline{K}(\m \beta): \m \beta \in \mathcal{B}_1 \cap \mathcal{B}_2, \Delta_{\m \beta} \ge 0\}
    \end{align}    
    leading to the following more conservative, compared to $K^{\ast \ast}$ in \eqref{gKss}, lower bound on the AV penetration rate:
    \begin{align} \label{UpbKst}
        & \boxed{\gamma \ge K^{\ast}(\hat{\m \beta}),}
    \end{align}
\end{myprs}
\begin{IEEEproof}
See Appendix \ref{App4}.
\end{IEEEproof}
Via nonlinear optimization techniques (by searching over the parameterized $\mathcal{H}_{\infty}$-based stabilizing control parameters), $K^{\ast}(\hat{\m \beta})$ in \eqref{UpbKst} can be computed.

\subsubsection{Less conservative lower bound on the AV penetration rate} \label{C3} Built upon the parameterized control parameter vector $\m \beta(\m \theta)$ characterized by \eqref{betap} in Corollary \ref{Cor1}, we construct the following feasibility problem:
    \begin{align}\label{Feas}
    \underset{\m \theta \in \mathbb{R}^3} {\mathrm{find}} \;\; \m \theta \;\; 
    \mathrm{subject~to} \;\;\rho[\m M(\m \beta(\m \theta))] \le 0.
\end{align}
To find a feasible solution $\m \theta$ for \eqref{Feas}, we can solve the following unconstrained optimization problem:
\begin{align} \label{UncOpt}
    & \underset{\m \theta \in \mathbb{R}^3}{\mathrm{minimize}}\;~\rho[\m M(\m \beta(\m \theta))],
\end{align}
for $\m \theta^{\ast}$ and if $\rho[\m M(\m \beta(\m \theta^{\ast}))] \le 0$ holds, then $\m \theta^{\ast}$ will be a feasible solution for the feasibility problem \eqref{Feas}. One can solve the unconstrained optimization problem \eqref{UncOpt} for $\m \theta^{\ast}$ via any nonlinear optimization solver. The optimal control parameter vector $\m \beta^{\ast}$ in Problem \ref{Prob1} is $\m \beta(\m \theta^{\ast})$.

Given a fixed number of heterogeneous HVs, namely $N_{\mathrm{HV}}$, according to \eqref{UpbKst} we equivalently get the following more conservative, compared to $\Big \lceil \frac{K^{\ast \ast}}{1-K^{\ast \ast}} N_{\mathrm{HV}} \Big \rceil$, lower bound on the number of required identical AVs to stabilize traffic flow:
\begin{align} \label{LNAV}
    & N_{\mathrm{AV}} \ge \bigg \lceil \frac{K^{\ast}(\hat{\m \beta})}{1-K^{\ast}(\hat{\m \beta})} N_{\mathrm{HV}} \bigg \rceil.
\end{align}
Starting from $m^{(0)} = \Big \lceil \frac{K^{\ast}(\hat{\m \beta})}{1-K^{\ast}(\hat{\m \beta})} N_{\mathrm{HV}} \Big \rceil$ and iteratively solving \eqref{UncOpt} for each $m$ along with a bisection method, we can verify if \eqref{Feas} admits a feasible solution $\m \theta^{\ast}$ or not. Via such a bisection-based iterative routine, we aim to compute the smallest $m$, namely $m^{\ast}$, for which \eqref{Feas} admits a feasible solution $\m \theta^{\ast}$. We then get the following less conservative lower bound on the AV penetration rate:
\begin{align} \label{gamstar}
    & \boxed{\gamma \ge \overbrace{\frac{m^{\ast}}{m^{\ast}+ N_{\mathrm{HV}}}}^{\gamma^{\ast}}.}
\end{align}
For the lower bounds on the AV penetration rate $\gamma$ derived by \eqref{gKss}, \eqref{UpbKst}, and \eqref{gamstar}, the following inequalities hold:
\begin{align} \label{MRes}
    & \boxed{\overbrace{K^{\ast}(\hat{\m \beta})}^{\text{More~Conservative}} \ge \overbrace{K^{\ast \ast}}^{\text{Conservative}} \ge \overbrace{\gamma^{\ast}}^{\text{Less~Conservative}}.}
\end{align}
Remarkably, for the case of traffic flow with heterogeneous HVs, we derive two additional lower bounds on the AV penetration rate $\gamma$ in \eqref{MRes}: $\gamma^{\ast}$ (less conservative) and $K^{\ast}(\hat{\m \beta})$ (more conservative), in addition to the conservative lower bound on the AV penetration rate $\gamma$, i.e., $K^{\ast \ast}$ similarly proposed by \cite{wu2018stabilizing,10753487} for the case of traffic flow with homogeneous HVs.

\section{Numerical Simulations} \label{Prcre}
This section presents Algorithm \ref{alg:LCLbf} to compute the less conservative lower bound on the AV penetration rate $\gamma$. Then, we assess the effectiveness of the theoretical results by conducting numerical simulations in MATLAB R2024a. In Algorithm \ref{alg:LCLbf}, since the objective function \eqref{UncOpt} involves the spectral abscissa, which is a non-convex and non-smooth function, we accordingly employ the MATLAB built-in nonlinear optimization solver $\texttt{fminsearch}()$ (developed based on the Nelder–Mead simplex method \cite{lagarias1998convergence}) to solve \eqref{UncOpt} for $\m \theta^{\ast}$ using the initial point $\m \theta_0$ as a randomly generated initial $\m \theta$. It is noteworthy that, given a non-convex and non-smooth optimization problem, attaining a globally optimal solution is not generally guaranteed. Since the search space associated with the optimization variable $\m \theta$ is three-dimensional, the optimization solver is quite efficient in terms of computational time. To test the theoretical results, built upon the utilized numerical setup associated with the OV-modeled HVs dynamics in \cite{zheng2020smoothing}, we consider two main scenarios: \textit{(i)} homogeneous HVs and \textit{(ii)} heterogeneous HVs. We choose the lower and upper bounds on the control parameters as $\m \beta^{l} = 0.8 \times \mathbf{1}_3$ and $\m \beta^{u} = 2 \times \mathbf{1}_3$. For our considered problem setup, we highlight that the \textit{stability invariance} property holds, i.e., the eigenmode spectrum of $\m M(\m \beta (\m \theta))$ is independent of the AV formation as the eigenmodes are the $2n$ roots of \eqref{PCR} which is agnostic to the AV formation. However, since \cite{li2022cooperative} suggests that an AV formation with a uniform distribution can potentially achieve a better $\mathcal{H}_2$ performance, we choose to use a uniform distribution for the choice of AV formation in our numerical simulations.

\setlength{\floatsep}{2pt}{
\begin{algorithm}[!ht]
\caption{\textbf{Less conservative lower bound on the AV penetration rate $\gamma$}}\label{alg:LCLbf}
\KwData{$N_{\mathrm{HV}}$, $\m \alpha_j$s for $j \in \mathcal{I}_{\mathrm{HV}}$, $\m \beta^{l}$, $\m \beta^{u}$}

\tcc{String stability verification}

$m \gets 1$

Construct $\rho[\m M(\m \beta(\m \theta))]$ via \eqref{betap}

Initialize $\m \theta$ with $\m \theta_0$

Solve \eqref{UncOpt} for $\m \theta^{\ast}$ using the initial point $\m \theta_0$

\eIf{$\rho[\m M(\m \beta(\m \theta^{\ast}))] \le 0$}
{
    $m^{\ast} \gets 1$
}{\tcc{Bisection method initialization}
Construct $\underline{K}(\m \beta)$ via \eqref{LowB} 

Compute $\hat{\m \beta}$ via \eqref{hatB}

Compute $K^{\ast}(\hat{\m \beta})$ via \eqref{SuffCon2}

$m^{(0)} \gets \Big \lceil \frac{K^{\ast}(\hat{\m \beta})}{1-K^{\ast}(\hat{\m \beta})} N_{\mathrm{HV}} \Big \rceil$ 

$a_0 \gets 1$, $b_0 \gets m^{(0)}$

\tcc{Bisection method}
\While{$b_i - a_i > 1$}{\tcc{String stability verification} $m^{(i)} \gets \big \lceil \frac{a_i+b_i}{2} \big \rceil$

Construct $\rho[\m M(\m \beta(\m \theta))]$ via \eqref{betap}

Solve \eqref{UncOpt} for $\m \theta^{\ast}$ using the initial point $\m \theta_0$

\eIf{$\rho[\m M(\m \beta(\m \theta^{\ast}))] \le 0$}{$a_{i+1} \gets a_i$, $b_{i+1} \gets m^{(i)}$}{$a_{i+1} \gets m^{(i)}$, $b_{i+1} \gets b_i$}

$m^{\ast} \gets b_{i+1}$}

}

Compute $\gamma^{\ast}$ via \eqref{gamstar}

\KwResult{$\gamma^{\ast}$}
\end{algorithm}}

\subsection{Homogeneous scenario}
We consider the following settings for all $j \in \mathcal{I}_{\mathrm{HV}}$: $\m \alpha_j = \begin{bmatrix} 0.3\pi & 1.5 & 0.9 \end{bmatrix}^\top,~\Delta_{\m \alpha_j} = -0.4450 < 0,$
with $\rho = 0.0256 > 0$, i.e., the uncontrolled platoon is string unstable. For this setup, applying the procedure proposed by \cite{10753487}, one can get $K^{\ast \ast} = 0.1541$. Running Algorithm \ref{alg:LCLbf} along with $N_{\mathrm{HV}} = 24$, we get $\m \beta(\m \theta^{\ast}) = \begin{bmatrix}
    0.8018 & 2.0000 & 0.8001
\end{bmatrix}^\top$ satisfying $\rho[\m M(\m \beta(\m \theta^{\ast}))] = -1.7183 \times 10^{-15} \le 0$, i.e., the controlled platoon is string stable, $m^{\ast} = 4$, and $\gamma^{\ast} = 0.1429$. Moreover, we get $\hat{\m \beta} = \begin{bmatrix}
    0.8000 & 2.0000 & 0.8000
\end{bmatrix}$ and $K^{\ast}(\hat{\m \beta}) = 0.1541$. We observe that $K^{\ast}(\hat{\m \beta}) = K^{\ast \ast} > \gamma^{\ast}$ holds. Since
\begin{align*}
    \overbrace{m^{\ast}}^{4} < \overbrace{\bigg \lceil \frac{K^{\ast \ast}}{1-K^{\ast \ast}} N_{\mathrm{HV}} \bigg \rceil}^{5} = \overbrace{\bigg \lceil \frac{K^{\ast}(\hat{\m \beta})}{1-K^{\ast}(\hat{\m \beta})} N_{\mathrm{HV}} \bigg \rceil}^{5}
\end{align*}
holds, the proposed direct method (less conservative) outperforms the $\mathcal{H}_{\infty}$-based alternative (conservative) proposed by \cite{10753487} in terms of conservativeness. Quantitatively, it attains $100 \times \frac{\frac{5}{5+24}-\frac{4}{4+24}}{\frac{5}{5+24}} = 17.14\%$ AV penetration rate reduction (improvement) while ensuring the string stability of traffic flow.

\subsection{Heterogeneous scenario}
We consider the following settings for all $j \in \mathcal{I}_{\mathrm{HV}}$: $\m \alpha_j = \begin{bmatrix} 0.3\pi & 1.5 & 0.9 \end{bmatrix}^\top \odot (\mathbf{1}_3 + \m \nu_j \odot \m \kappa),$
where $\m \nu_j$s are $3$-dimensional randomly generated vectors via $2*\texttt{rand}(3,1)-\texttt{ones}(3,1)$ and $\m \kappa \in \mathbb{R}_{+}^3$ denotes the deviation parameter. In other terms, we implement the heterogeneity of HVs via randomly generated $\m \nu_j$s. Also, we have $\rho = 0.0347 > 0$, i.e., the uncontrolled platoon is string unstable. Running Algorithm \ref{alg:LCLbf} along with $N_{\mathrm{HV}} = 24$ and $\m \kappa = 0.15 \times \mathbf{1}_3$, we get $\m \beta(\m \theta^{\ast}) = \begin{bmatrix}
    0.8000 & 2.0000 & 0.8009
\end{bmatrix}^\top$ satisfying $\rho[\m M(\m \beta(\m \theta^{\ast}))] = -1.7867 \times 10^{-15} \le 0$, i.e., the controlled platoon is string stable, $m^{\ast} = 5$, and $\gamma^{\ast} = 0.1724$. Moreover, we get $\hat{\m \beta} = \begin{bmatrix}
    0.8000 & 2.0000 & 0.8000
\end{bmatrix}$ and $K^{\ast}(\hat{\m \beta}) = 0.1813$. We observe that $K^{\ast}(\hat{\m \beta}) > \gamma^{\ast}$ holds. Since
\begin{align*}
    \overbrace{m^{\ast}}^{5} < \overbrace{\bigg \lceil \frac{K^{\ast}(\hat{\m \beta})}{1-K^{\ast}(\hat{\m \beta})} N_{\mathrm{HV}} \bigg \rceil}^{6}
\end{align*}
holds, the proposed direct method (less conservative) outperforms the proposed $\mathcal{H}_{\infty}$-based alternative (more conservative) in terms of conservativeness. As previously mentioned, computing $K^{\ast \ast}$ is computationally challenging. However, according to \eqref{MRes}, we realize that $K^{\ast \ast} \in [0.1724,0.1813]$ holds.

To visualize the effects of two string stability criteria (direct and $\mathcal{H}_{\infty}$-based), we run Algorithm \ref{alg:LCLbf} along with $N_{\mathrm{HV}} = 12$ and $\m \kappa = 0.05 \times \mathbf{1}_3$. Also, we have $\rho = 0.0038 > 0$, i.e., the uncontrolled platoon is string unstable. Then, we get $\m \beta(\m \theta^{\ast}) = \begin{bmatrix}
    1.0329 & 1.9860 & 0.8258
\end{bmatrix}^\top$, satisfying $\rho[\m M(\m \beta(\m \theta^{\ast}))] = -1.2210 \times 10^{-15} \le 0$, $m^{\ast} = 1$, and $\gamma^{\ast} = 0.0769$, i.e., the platoon controlled by a single identical AV is string stable. Moreover, we get $\hat{\m \beta} = \begin{bmatrix}
    0.8000 & 2.0000 & 0.8000
\end{bmatrix}$ and $K^{\ast}(\hat{\m \beta}) = 0.1528$. We observe that $K^{\ast}(\hat{\m \beta}) > \gamma^{\ast}$ holds. We also get
\begin{align*}
    \overbrace{m^{\ast}}^{1} < \overbrace{\bigg \lceil \frac{K^{\ast}(\hat{\m \beta})}{1-K^{\ast}(\hat{\m \beta})} N_{\mathrm{HV}} \bigg \rceil}^{3}.
\end{align*}
Similarly, according to \eqref{MRes}, we realize that $K^{\ast \ast} \in [0.0769,0.1528]$ holds.

For $m = \Big \lceil \frac{K^{\ast}(\hat{\m \beta})}{1-K^{\ast}(\hat{\m \beta})} N_{\mathrm{HV}} \Big \rceil = 3$, we get $\m \beta(\m \theta^{\ast}) = \begin{bmatrix}
    0.8070 & 1.9621 & 1.0171
\end{bmatrix}^\top$, satisfying $\rho[\m M(\m \beta(\m \theta^{\ast}))] = -1.3264 \times 10^{-15} \le 0$, i.e., the controlled platoon is string stable. Fig. \ref{Fig1} depicts the corresponding eigenmode visualizations for the less conservative and more conservative scenarios. Running the simulations with an initial perturbation of magnitude $1$ at the $n$-th vehicle's location (heterogeneous HV), Fig. \ref{Fig1} depicts the location deviation trajectories of the $(n-1)$-th and the $2$-nd vehicles (heterogeneous HVs) and the $1$-st vehicle (identical AV) for the less and more conservative scenarios. Comparing the two plots associated with $j= n-1$ (HV) on the left with the two plots associated with $j=2$ on the right, we observe that AV $j=1$ has effectively dampened the amplified perturbation through the platoon (originating at the $n$-th vehicle propagating from vehicle $j = n-1$ through vehicles $j= n-2,\dots,3$ to $j=2$). 

\begin{figure*}[t]
    \centering
    \includegraphics[scale=0.27]{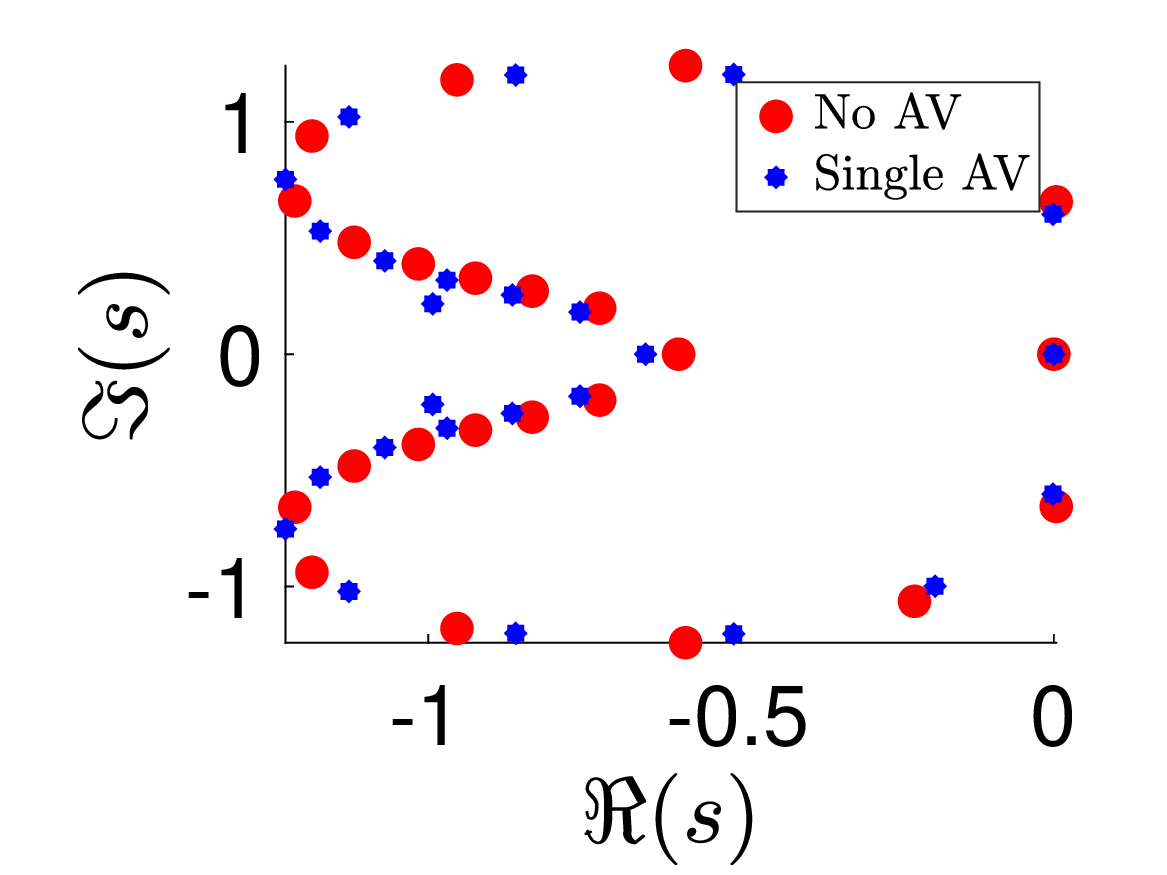}
    \includegraphics[scale=0.27]{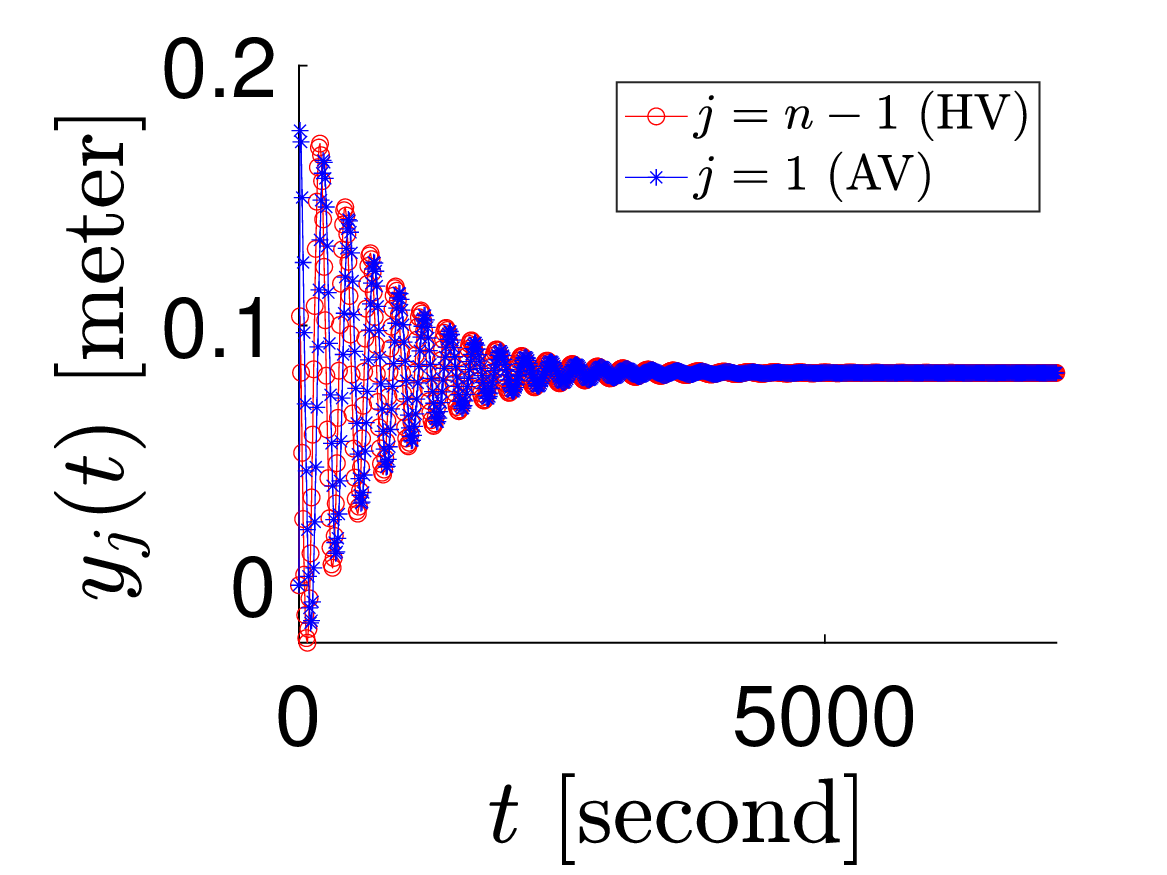}
    \includegraphics[scale=0.27]{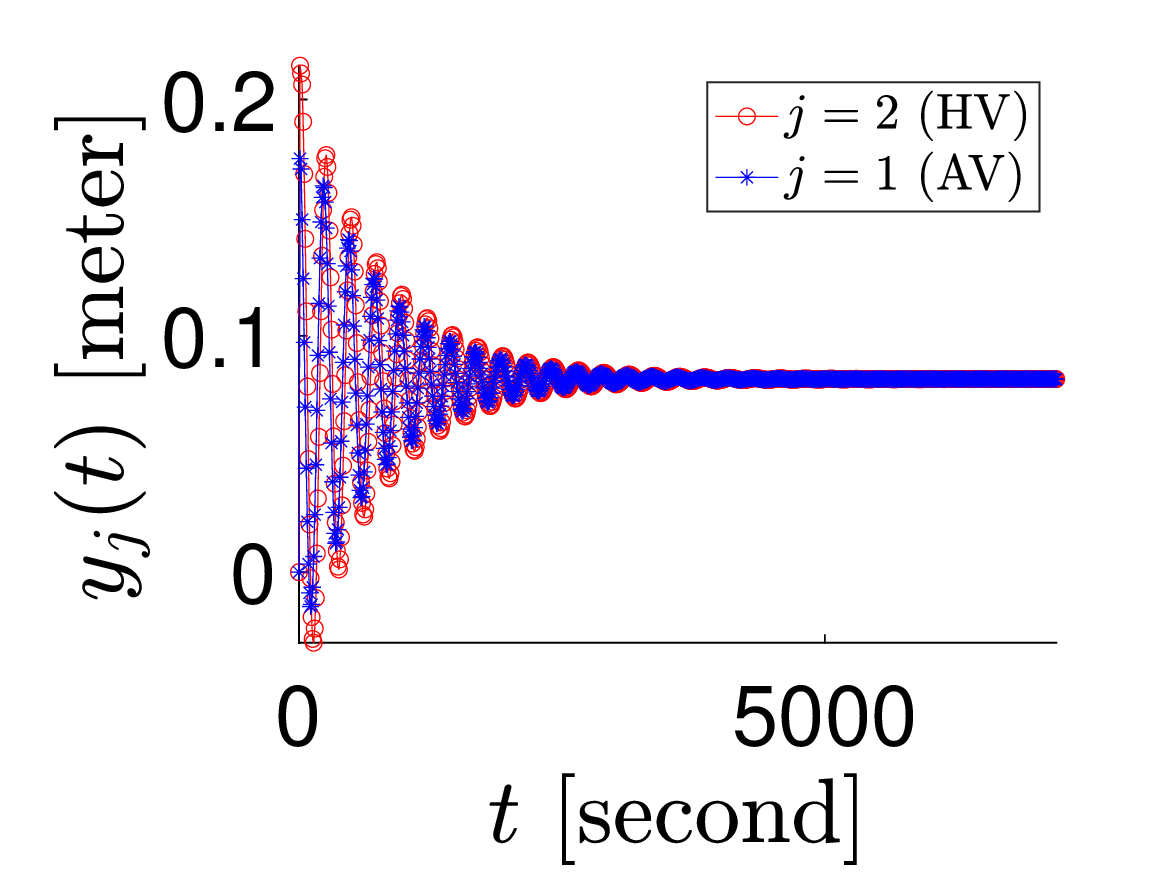}
    
    \includegraphics[scale=0.27]{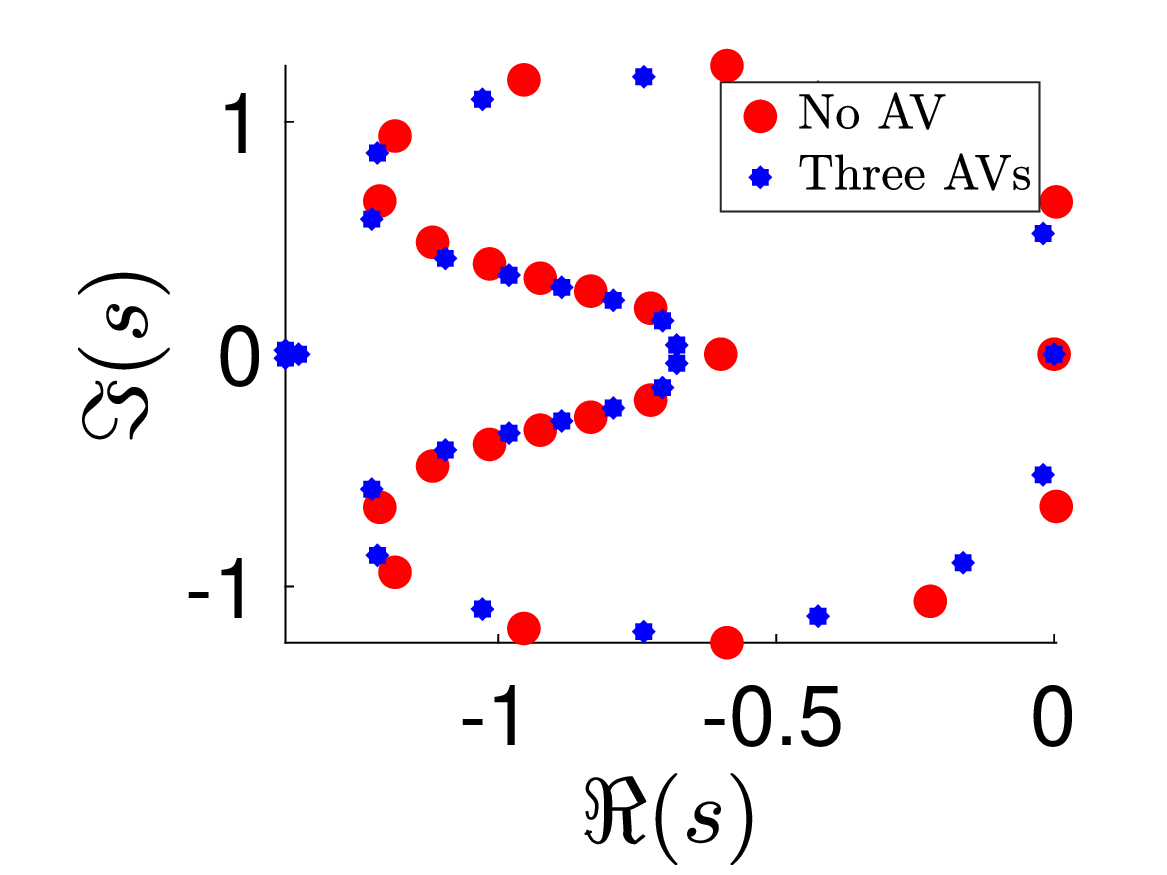}
    \includegraphics[scale=0.27]{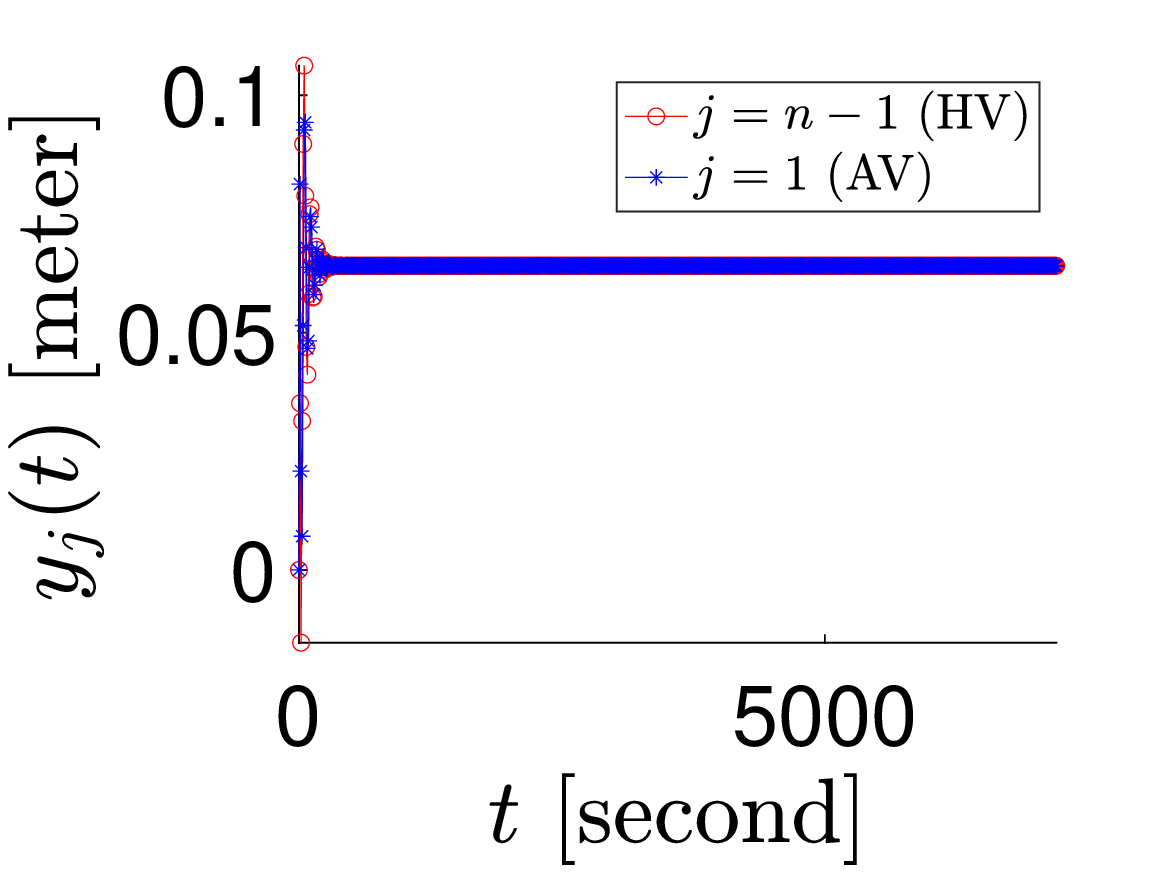}
    \includegraphics[scale=0.27]{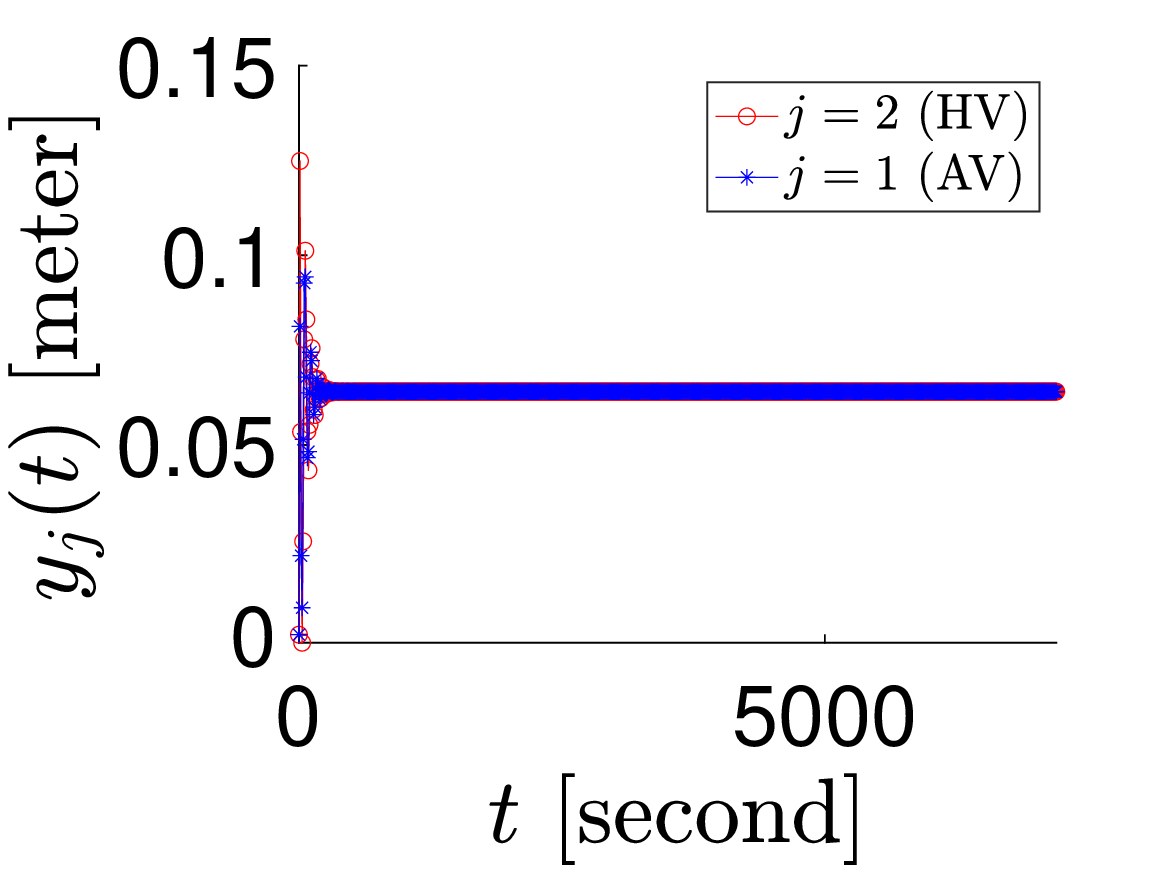}
    
    \caption{Eigenmode visualizations and location deviation trajectories of vehicles $j=n-1$ and $j=2$ (heterogeneous HVs) and $j=1$ (AV). The top row shows $m=m^{\ast}=1$, and the bottom row shows $m=\lceil K^{\ast}(\hat{\m\beta})N_{\mathrm{HV}}/(1-K^{\ast}(\hat{\m\beta}))\rceil=3$. Here, $m$ is the number of identical AVs, $m^{\ast}$ is the smallest $m$ for which \eqref{Feas} admits a feasible solution $\m\theta^{\ast}$, and $N_{\mathrm{HV}}=12$ is the fixed number of heterogeneous HVs. The quantities $K^{\ast}(\m\beta)$ and $\hat{\m\beta}$ are defined in \eqref{SuffCon2} and \eqref{hatB}, respectively.}
    \label{Fig1}
\end{figure*}

As Fig. \ref{Fig1} demonstrates, both less and more conservative scenarios successfully stabilize the traffic flow while attaining acceptable location deviation performance. However, we notice that for the more conservative scenario $m = \Big \lceil \frac{K^{\ast}(\hat{\m \beta})}{1-K^{\ast}(\hat{\m \beta})} N_{\mathrm{HV}} \Big \rceil = 3$, the closed-loop eigenmodes become more distant from the imaginary axis. Moreover, the more conservative scenario $m = \Big \lceil \frac{K^{\ast}(\hat{\m \beta})}{1-K^{\ast}(\hat{\m \beta})} N_{\mathrm{HV}} \Big \rceil = 3$ attains relatively less (better) location deviation. However, these are all achieved at the expense of the utilization of $2$ more identical AVs. In other terms, a trade-off exists between the stabilization/performance degradation and the number of utilized identical AVs. Quantitatively, the data associated with Fig. \ref{Fig1} corroborate that the AV penetration rate can be reduced by $100 \times \frac{\frac{3}{3+12}-\frac{1}{1+12}}{\frac{3}{3+12}} = 61.54\%$ at the expense of $100 \times \frac{0.0817-0.0640}{0.0640} = 27.66\%$ higher position difference deviation from the equilibrium while ensuring the string stability of traffic flow. The information extracted from such a trade-off can benefit traffic control engineers/operators in traffic control decision-making tasks.

To quantify the \textit{distance to instability} associated with the state matrix $\m M(\m \beta(\m \theta^{\ast}))$ of the aggregated linearized dynamics (after the dimension reduction through eliminating the zero eigenmode as discussed later on), we utilize the notion of \textit{minimum destabilizing real perturbation (MDRP)} \cite{van1984near}. The MDRP is also known as the real stability radius (RSR) \cite{hinrichsen1986stability}. The RSR is a metric to measure the stability robustness under the perturbation/uncertainty. The RSR for a given stable matrix $\m W \in \mathbb{R}^{n_W \times n_W}$ (i.e., $\rho[\m W + \m X] < 0$) is mathematically defined as follows:
\begin{align*}
    & \mathrm{RSR}(\m W) := \min \{\|\m X\|_F: \rho[\m W + \m X] = 0, \m X \in \mathbb{R}^{n_W \times n_W}\}. 
\end{align*}
The exact computation of the RSR is computationally intractable (NP-hard), and no practical estimation
technique could be devised due to the difficulty of the ensuing constrained minimization problem cast by \cite{van1984near}. Thus, to approximately compute the RSR, we utilize a hybrid expansion-contraction (HEC)-based method proposed by \cite{guglielmi2017approximating} that has been implemented in the robust stability package, namely ROSTAPACK, developed by \cite{Mitchell2022RObust}. Precisely, the corresponding command is called \texttt{getStabRadBound}.

According to the structure of the $\m \beta$-dependent matrices $\m A(\m \beta) \in \mathbb{R}^{n \times n}$ and $\m B(\m \beta) \in \mathbb{R}^{n \times n}$ defined by \eqref{LD}, it can be verified that $\m A(\m \beta) \mathbf{1}_n = \mathbf{0}_n$ holds and as a result, $\m M(\m \beta)$ satisfies the following eigenmode equation:
\begin{align*}
    & \m M(\m \beta) \begin{bmatrix}
        \mathbf{1}_n \\ \mathbf{0}_n
    \end{bmatrix} = 0 \times \begin{bmatrix}
        \mathbf{1}_n \\ \mathbf{0}_n
    \end{bmatrix},
\end{align*} indicating that $\m M(\m \beta)$ has an eigenvector $\begin{bmatrix}
        \mathbf{1}_n \\ \mathbf{0}_n
    \end{bmatrix}$ associated with a $0$ eigenmode. To eliminate the $0$ eigenmode from $\m M(\m \beta)$ and obtain a reduced matrix, we employ a linear algebraic---similarity transformation-based approach---similar to the one employed by \cite{bahavarnia2025quick} as follows:
    \begin{align*}
        & \m M_{\mathrm{reduced}}(\m \beta) = \m T^\top \m M(\m \beta) \m T,\\
        & \m T = \begin{bmatrix}
            \m U & \m O_{n \times n}\\
            \m O_{n \times (n-1)} & \m I_n
        \end{bmatrix},~\mathrm{Similarity~Transformation},\\
        & \m U :~\mathrm{A~Matrix~consisting~of~the~Concatenated}\\
        &\mathrm{Eigenvectors~associated~with~the~Nonzero~Eigenmodes~of}\\&\m I_n - \frac{1}{n} \mathbf{1}_n \mathbf{1}_n^\top. 
    \end{align*} Now, we are ready to compute the RSR of the stable reduced matrix $\m M_{\mathrm{reduced}}(\m \beta(\m \theta^{\ast}))$ via ROSTAPACK \cite{Mitchell2022RObust}. Quantitatively, the data associated with Fig. \ref{Fig1} corroborate that the aforementioned AV penetration rate reduction (improvement) of $61.54\%$ can be achieved at the cost of $100 \times \frac{0.0186-0.0014}{0.0186} = 92.47\%$ degradation in the RSR of the corresponding stable reduced matrix $\m M_{\mathrm{reduced}}(\m \beta(\m \theta^{\ast}))$. Remarkably, the RSR is a lower bound on the \textit{structured} RSR, which is defined subject to the following reduced, structured perturbation/uncertainty:
\begin{align*}
    & X_{\mathrm{reduced}}^{\mathrm{structured}} = \m T^\top \begin{bmatrix}
        \m O_{n \times n} & \m O_{n \times n}\\
        \m \delta \m A & \m \delta \m B
    \end{bmatrix} \m T,\\
    & \m \delta \m A :
       {\delta a}_{jj} = -{\delta \alpha}_{j1},~{\delta a}_{j(j+1)} = {\delta \alpha}_{j1} \quad j \in \mathcal{I}_{\mathrm{HV}},\\
    & \m \delta \m B:
       {\delta b}_{jj} = -{\delta \alpha}_{j2},~{\delta b}_{j(j+1)} = {\delta \alpha}_{j3} \quad j \in \mathcal{I}_{\mathrm{HV}},
\end{align*}where the perturbation/uncertainty vectors ${\m \delta \m \alpha}_j = \begin{bmatrix} {\delta \alpha}_{j1} & {\delta \alpha}_{j2} & {\delta \alpha}_{j3}
\end{bmatrix}^\top$ for $j \in \mathcal{I}_{HV}$ could model the norm-bounded linearization error, computational inaccuracies, or any other source of variation. Since the computation of the structured RSR is computationally cumbersome, we chose the RSR over the structured RSR for the purpose of analyzing stability robustness degradation.

\section{Concluding Remarks} \label{Con}
In this paper, we show that one can attain a less conservative string stable traffic flow via identical AVs. To that end, we ensure the string stability of traffic flow by directly imposing the possession of no growing eigenmodes. Moreover, given a fixed number of heterogeneous HVs, built upon the derived theoretical lower bounds, we systematically find a minimum number of required identical AVs and solve for the optimal control parameters via nonlinear optimization techniques. Precisely, for the case of traffic flow with heterogeneous HVs, we derive three theoretical lower bounds on the AV penetration rate $\gamma$: \textit{(i)} conservative, \textit{(ii)} more conservative, and \textit{(iii)} less conservative. Since computing the conservative lower bound is computationally challenging, we use the more conservative lower bound as a basis to initialize the bisection-based iterative routine---as a crucial part of the main algorithm---specialized for computation of the less conservative lower bound. Finally, through numerical simulations, we verify the effectiveness of the identical AV-based perturbation attenuation---string stability---and capture a trade-off between the stabilization/performance degradation and the number of utilized identical AVs. The information extracted from such a trade-off can be beneficial to traffic control engineers/operators.

\textit{Future Directions:} As pertinent future work, one can think of the investigation of the non-identical AVs scenario and potentially the extraction of the new achievable stability/performance goals via non-identical AVs---with more free control parameters. As other possible future work, a thorough study of a more realistic noisy scenario seems paramount. Moreover, in the future, a thorough trade-off analysis between safety, efficiency (referred to as mobility), and stability can be conducted to shed light on an additional aspect of the conservativeness.

\appendices

\section{Proof of Proposition \ref{Propo1}} \label{App1}
Imposing BCs \eqref{LUB} on the parameterization \eqref{pars}, we get
\begin{subequations} \label{params}
    \begin{align}
        & \beta_3^l \le p \le \beta_3^u,\\
        & \beta_2^l \le p+q \le \beta_2^u,\\
        & \beta_1^l \le r \le \beta_1^u.
    \end{align}
\end{subequations}
We also consider $p > 0$, $q > 0$, and $r > 0$ as $p \ge \epsilon$, $q \ge \epsilon$, and $r \ge \epsilon$, respectively, where $\epsilon > 0$ denotes an infinitesimal value. Then, \eqref{params} along with $p \ge \epsilon$, $q \ge \epsilon$, and $r \ge \epsilon$ implies that
\begin{subequations} \label{pqr}
    \begin{align}
        \epsilon \le & p, \label{pqr1}\\
        \beta_3^l \le & p, \label{pqr2}\\
        & p\le \beta_3^u, \label{pqr3}\\
        & p \le \beta_2^u-\epsilon, \label{pqr4}\\
        \epsilon \le & q, \label{pqr5}\\
        \beta_2^l-p \le & q, \label{pqr6}\\
        & q\le \beta_2^u-p, \label{pqr7}\\
        \epsilon \le & r, \label{pqr8}\\
        \beta_1^l \le & r, \label{pqr9}\\
        & r \le \beta_1^u, \label{pqr10}
    \end{align}
\end{subequations}
hold. Notice that \eqref{pqr4} is obtained from the combination of \eqref{pqr5} and \eqref{pqr7}. Thus, utilizing \eqref{pqr} in conjunction with the convex combination notion, the $p$, $q$, and $r$ in \eqref{pars} satisfying the BCs \eqref{LUB} can be parameterized as \eqref{pqrpar}. Also, we have
\begin{align*}
    & \eqref{pqr1},\eqref{pqr3} \implies \beta_3^u \ge \epsilon,\\
    & \eqref{pqr1},\eqref{pqr4} \implies \beta_2^u \ge 2\epsilon,\\
    & \eqref{pqr2},\eqref{pqr4} \implies \beta_2^u \ge \beta_3^l + \epsilon,\\
    & \eqref{pqr8},\eqref{pqr10} \implies \beta_1^u \ge \epsilon,
\end{align*}
which completes the proof.

\section{Proof of Lemma \ref{Lemma1}} \label{App2}
To find the globally maximizer of $D_{\m \alpha_j}(\omega)$, we must solve $\frac{d D_{\m \alpha_j}(\omega)}{d\omega} = 0$ for $\tilde{\omega}_j$. Since $D_{\m \alpha_j}(\omega)$ is continuous on $\Big[0,\sqrt{-\Delta_{\m \alpha_j}}\Big]$ and differentiable on $\Big(0,\sqrt{-\Delta_{\m \alpha_j}}\Big)$, according to the Rolle's theorem \cite{thomas1992calculus} $D_{\m \alpha_j}(0) = D_{\m \alpha_j}\Big(\sqrt{-\Delta_{\m \alpha_j}}\Big) = 0$ implies that there exists at least one point $\tilde{\omega}_j \in \Big(0,\sqrt{-\Delta_{\m \alpha_j}}\Big)$ for which $\frac{d D_{\m \alpha_j}(\omega)}{d\omega} = 0$ is satisfied. Considering $\frac{d D_{\m \alpha_j}(\omega)}{d\omega} = 0$, we get
\begin{align} \label{Der0}
    & \frac{2\alpha_{j3}^2 \omega}{\alpha_{j3}^2 \omega^2 + \alpha_{j1}^2} - \frac{2 \omega \alpha_{j2}^2 + 4 \omega (\omega^2-\alpha_{j1})}{\alpha_{j2}^2 \omega^2 + (\omega^2-\alpha_{j1})^2} = 0.
\end{align}
For $\omega \neq 0$, after some mathematical manipulations, \eqref{Der0} reduces to the following equation:
\begin{align} \label{Eqw}
    & \alpha_{j3}^2 \omega^4 + 2 \alpha_{j1}^2 \omega^2 + \alpha_{j1}^2 \Delta_{\m \alpha_j} = 0.
\end{align}
Solving \eqref{Eqw} for $\tilde{\omega}_j > 0$ and noting that $\frac{\Delta_{\m \alpha_j}}{\alpha_{j3}^2} < 0$ and $-\frac{2 \alpha_{j1}^2}{\alpha_{j3}^2} < 0$ hold, we get the unique $\tilde{\omega}_j$ in \eqref{tilomeg}. Because $D_{\m \alpha_j}(\omega) > 0$ holds for all $\omega \in \Big(0,\sqrt{-\Delta_{\m \alpha_j}}\Big)$ according to \eqref{SignDforDeltaalpha}, such a unique $\tilde{\omega}_j$ is a global maximizer. Moreover, by evaluating $D_{\m \alpha_j}(\omega)$ at $\tilde{\omega}_j$, we obtain \eqref{GMV}.

\section{Proof of Proposition \ref{Propo2}} \label{App3}
By the definition of $J^{\ast}(\m \beta)$ in \eqref{SuffCon2}, we have
\begin{align} \label{JDD}
    \frac{-D_{\m \beta}(\tilde{\omega}_j)}{D^{\mathrm{AM}}(\tilde{\omega}_j)} \ge J^{\ast}(\m \beta),~\forall j \in \tilde{\mathcal{I}}^{+}_{\mathrm{HV}}.
\end{align}
Taking the minimum from the RHS of \eqref{JDD} completes the proof.

\section{Proof of Proposition \ref{Propo3}} \label{App4}
Let us define $\m \beta^{\min} := \arg \min \{K^{\ast}(\m \beta): \m \beta \in \mathcal{B}_1 \cap \mathcal{B}_2, \Delta_{\m \beta} \ge 0\}$ for which $K^{\ast}(\m \beta^{\min}) = K^{\ast \ast}$ holds. Then, by definition $K^{\ast}(\m \beta) \ge K^{\ast \ast}$ holds for any $\m \beta$ belonging to $\mathcal{B}_1 \cap \mathcal{B}_2 \cap \{\m \beta \in \mathbb{R}_{+}^{3}: \eqref{SuffCon1}\mathrm{~holds~for~}\m \beta\}$. Then, choosing $\hat{\m \beta}$, we get $K^{\ast}(\hat{\m \beta}) \ge K^{\ast \ast}$. We obtain a more conservative (compared to $K^{\ast \ast}$) lower bound on the AV penetration rate by substituting $\hat{\m \beta}$ in \eqref{SuffCon2}, which completes the proof.

\bibliographystyle{IEEEtran}
\bibliography{References_final}

\vspace{-1cm}
\begin{IEEEbiography}[{\includegraphics[width=1in,height=1.25in,clip,keepaspectratio]{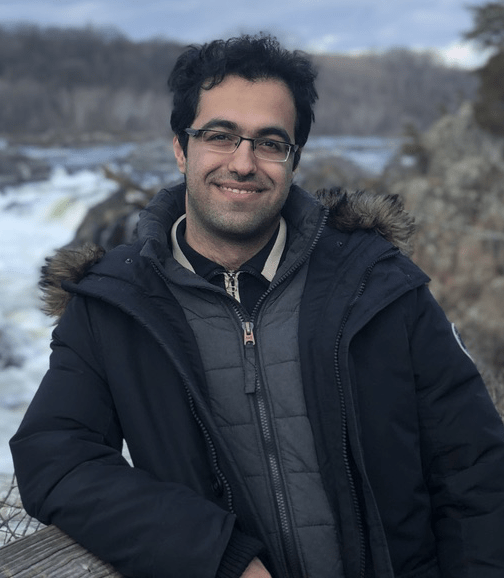}
 }] {MirSaleh Bahavarnia} (Member, IEEE) received a B.Sc. degree in Electrical Engineering (Control) and a certificate from the minor program in Mathematics from the Sharif University of Technology, Tehran, Iran, in 2013 and a Ph.D. degree in Mechanical Engineering (Control) from Lehigh University, Bethlehem, PA, USA, in 2018. He was a Postdoctoral Research Associate with the Department of Electrical and Computer Engineering and the Institute for Systems Research (ISR), University of Maryland, College Park, MD, USA, from 2018 to 2020. He was a Postdoctoral Research Scholar with the Department of Civil and Environmental Engineering, Vanderbilt University, Nashville, TN, USA, from 2022 to 2025. Since 2026, he has been a research scientist with the Department of Civil and Environmental Engineering, Vanderbilt University, Nashville, TN, USA. His research interests include distributed control, feedback control, power systems control, process control, robust control, and traffic control.
\end{IEEEbiography}

\vspace{-1cm}
\begin{IEEEbiography}[{\includegraphics[width=1in,height=1.25in,clip,keepaspectratio]{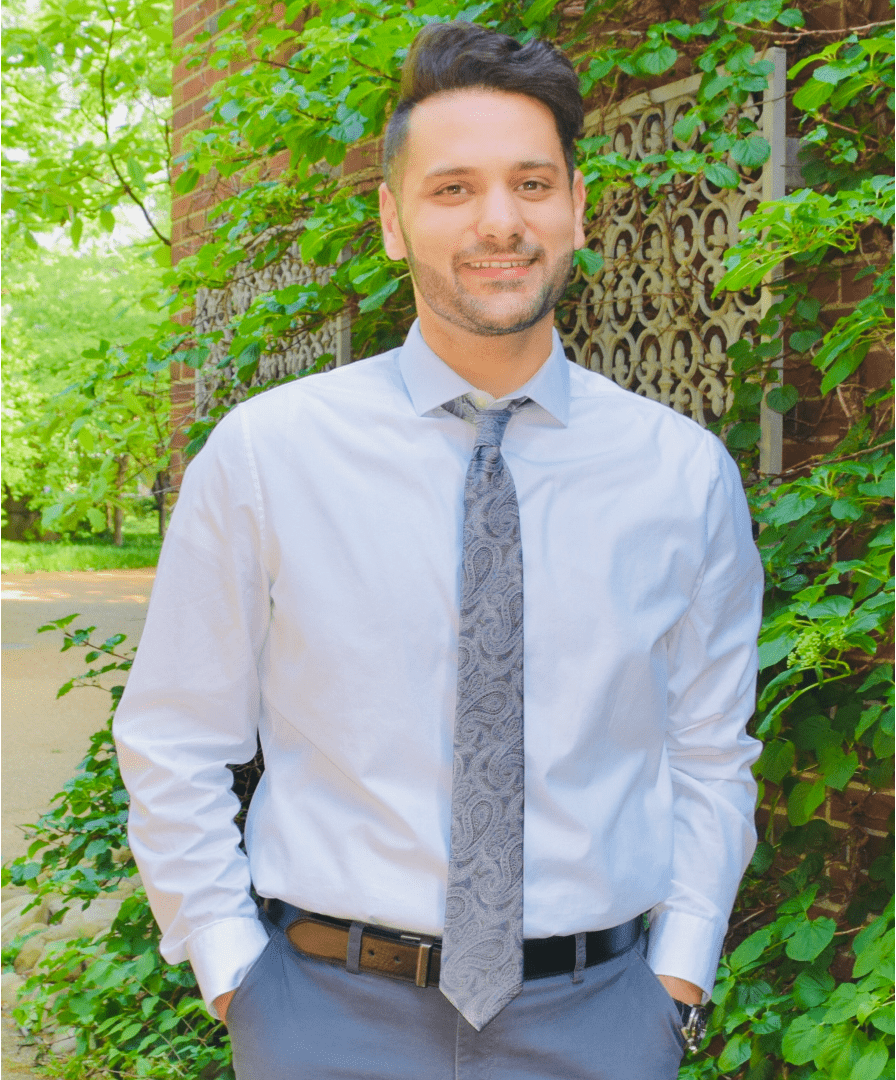}
 }] {Ahmad F. Taha} (Member, IEEE) received the B.E. degree in electrical and computer engineering from the American University of Beirut, Beirut, Lebanon, in 2011, and the Ph.D. degree in electrical and computer engineering from Purdue University, West Lafayette, IN, USA, in 2015. Before joining Vanderbilt University, Nashville, TN, USA, he was an Assistant Professor with the ECE Department at the University of Texas, San Antonio. He is an Associate Professor with the Department of Civil and Environmental Engineering at Vanderbilt University. He has a secondary appointment in the ECE Department. His research interests include understanding how complex cyber-physical and urban infrastructures operate, behave, and occasionally misbehave, and optimization, control, monitoring, and security of infrastructure with power, water, and transportation systems applications. Dr. Taha is an Associate Editor for the \textsc{IEEE Transactions on Control of Network Systems}.

\end{IEEEbiography}

\end{document}